\documentclass[12pt]{article}
\usepackage{amsmath}
\usepackage{amsthm}
\usepackage{amssymb}

\newtheorem{theorem}{Theorem}

\usepackage{times}
\usepackage{graphicx}
\usepackage{color}
\usepackage{multirow}
\usepackage[authoryear]{natbib}
\usepackage{rotating}
\usepackage{bbm}
\usepackage{latexsym}

\renewcommand{\baselinestretch}{1.2}
\newcommand{\captionfonts}{\normalsize}

\makeatletter  
\long\def\@makecaption#1#2{%
  \vskip\abovecaptionskip
  \sbox\@tempboxa{{\captionfonts #1: #2}}%
  \ifdim \wd\@tempboxa >\hsize
    {\captionfonts #1: #2\par}
  \else
    \hbox to\hsize{\hfil\box\@tempboxa\hfil}%
  \fi
  \vskip\belowcaptionskip}
\makeatother   

\newcommand{\mysps}{\hspace{-0.05em}}
\newcommand{\myspl}{\hspace{-0.15em}}
\makeatletter
\DeclareRobustCommand{\rs}[1]{%
  \@tfor\next:=#1\do{\next\mysps}%
}
\makeatother
\makeatletter
\DeclareRobustCommand{\rl}[1]{%
  \@tfor\next:=#1\do{\next\myspl}%
}
\makeatother
\newcommand{\refeq}[1]{{Eq~(\ref{#1})}}

\newcommand{\reffig}[2]{{Fig~\ref{#1}{{#2}}}}

\newcommand{\refsfig}[2]{{Fig~\ref{#1}{{#2}}}}

\newcommand{\infoe}{I}
\newcommand{\infoes}[1]{\infoe_{#1}}
\newcommand{\infodq}[2]{I_{#1}^{\rl{#2}}}

\newcommand{\Ad}[2]{A_{#1}^{#2}}
\newcommand{\niexp}{{\rl{N}I}}
\newcommand{\niexps}{{\rl{NI}}}
\newcommand{\dienc}[2]{\Delta \infoe_{#2}}

\newcommand{\diae}[2]{\Delta A_{#2}}
\newcommand{\diaq}[2]{\Delta A_{}^{\rl{#2}}}

\newcommand{\diq}[2]{\Delta \infodq{}{#2}}
\newcommand{\VR}{\mathbf{R}}
\newcommand{\Vr}{\mathbf{r}}
\newcommand{\VRNI}{\mathbf{R}^\niexp}

\newcommand{\RT}{\hat{R}}
\newcommand{\rT}{\hat{r}}
\newcommand{\VRT}{\mathbf{\RT}}
\newcommand{\VrT}{\mathbf{\rT}}
\newcommand{\RB}{\breve{R}}
\newcommand{\VRB}{\mathbf{\RB}}
\newcommand{\rB}{\breve{r}}
\newcommand{\VrB}{\mathbf{\rB}}
\newcommand{\VS}{S}
\newcommand{\VST}{\hat{S}}
\newcommand{\VSTL}{\mathbf{\VST}}

\newcommand{\VSNI}{\VS^\niexp}
\newcommand{\VSNIL}{\mathbf{S}^\niexp}
\newcommand{\Vs}{s}
\newcommand{\VsT}{\hat{s}}

\usepackage{pgf}
\usepackage{tikz}

\usetikzlibrary{fit,shapes.misc}

\newcommand\acircled[1]{%
	\tikz[baseline=(X.base)] 
	\node (X) [draw, shape=circle, inner sep=-.7pt] {#1};}
\newcommand\asquared[1]{%
	\tikz[baseline=(X.base)] 
	\node (X) [draw, shape=rectangle, inner sep=0.4pt] {#1};}

\newcommand{\SCirc}{\acircled{\makebox[\widthof{A}][c]{\phantom{A}}}}
\newcommand{\SBox}{\asquared{\makebox[\widthof{A}][c]{\phantom{A}}}}
\newcommand{\SACirc}{\acircled{\makebox[\widthof{A}][c]{A}}}
\newcommand{\SABox}{\asquared{\makebox[\widthof{A}][c]{A}}}
\newcommand{\SBCirc}{\acircled{\makebox[\widthof{A}][c]{B}}}
\newcommand{\SBBox}{\asquared{\makebox[\widthof{A}][c]{B}}}

\newcommand{\PNI}{P^\niexp}
\newcommand{\PNIs}{P^{\niexps\mysps}}
\newlength{\mywidthalign}
\newcommand{\mboxalign}[1]{\makebox[\mywidthalign][l]{#1}}

\usepackage{calc}
\usepackage{multirow}
\everydisplay{\def\arraystretch{0.5}}

\begin{document}

\begin{center}
{\LARGE Neural Stochastic Codes, Encoding and Decoding}\\[3ex]
\textbf{\large Hugo Gabriel Eyherabide}\\[3ex]
\end{center}

{
\linespread{1}\selectfont
\begin{center}
\textit{Department of Computer Science and Helsinki Institute for Information Technology, University of Helsinki \linebreak Gustaf H\"allstr\"omin katu 2b, FI00560, Helsinki, Finland}\\[3ex]
\end{center}
}

\begin{center}
{\textbf{Email:} neuralinfo@eyherabidehg.com}\hfill {\textbf{Weppage:} eyherabidehg.com}
\end{center}


\begin{center} {\bf \large Abstract} \end{center}
Understanding brain function, constructing computational models and engineering neural prosthetics require assessing two problems, namely encoding and decoding, but their relation remains controversial. For decades, the encoding problem has been shown to provide insight into the decoding problem, for example, by upper bounding the decoded information. However, here we show that this need not be the case when studying response aspects beyond noise correlations, and trace back the actual causes of this major departure from traditional views. To that end, we reformulate the encoding and decoding problems from the observer or organism perspective. In addition, we study the role of spike-time precision and response discrimination, among other response aspects, using stochastic transformations of the neural responses, here called stochastic codes. Our results show that stochastic codes may cause different information losses when used to describe neural responses and when employed to train optimal decoders. Therefore, we conclude that response aspects beyond noise correlations may play different roles in encoding and decoding. In practice, our results show for the first time that decoders constructed low-quality descriptions of response aspects may operate optimally on high-quality descriptions and vice versa, thereby potentially yielding experimental and computational savings, as well as new opportunities for simplifying the design of computational brain models and neural prosthetics.

\pagebreak

\section{Introduction}

Assessing what aspects of neural activity are informative ---the encoding problem--- and what role they play in brain computations ---the decoding problem--- are two of the most fundamental questions in neuroscience, but their relation remain controversial. Traditionally, these two problems have been related through data processing theorems which ensure, for example, that the decoded information cannot exceed the encoded information \citep{schneidman2003,cover2006,victor2008,eyherabide2013}. Although seemingly valid when studying specific response aspects such as first-spike latencies and spike counts \citep{furukawa2002,nelken2007,gaudry2008,eyherabide2008,eyherabide2009,eyherabide2010JPP, eyherabide2010FCN}, this relation has recently been shown invalid when studying covariations between response aspects under repeated stimulation, also called noise correlations. Contrary to previously thought, here we show that this relation is also invalid when studying response aspects beyond noise correlations, and establish the actual conditions under which the relation holds.

Specifically, previous studies have shown that, although not always \citep{pillow2008,eyherabide2009}, decoding algorithms that ignore noise correlations may cause negligible information losses even in situations within which noise correlation have been found to increase the encoded information \citep{nirenberg2001,latham2005,oizumi2009,ince2010,oizumi2010,pita2011,meytlis2012,eyherabide2013,eyherabide2016,oizumi2016}. This paradoxical observation has previously been attributed to at least four different causes: inherent limitations of the employed measures, putative fundamental differences between ignoring response aspects and response probabilities, neural activity that would only occur should neurons fire independently but not in the recorded data, or the fact that, unlike latencies or spike counts, noise correlations cannot be ignored using deterministic functions of the recorded neural responses, also called reduced codes \citep{nirenberg2003,schneidman2003,latham2005,averbeck2006}. Although still subject to active research \citep{oizumi2009,ince2010,oizumi2010,eyherabide2013,eyherabide2016}, this paradox and its causes continue to be regarded as a sole hallmark of studying noise correlations.

This belief notwithstanding, here we show that analogous observations arise when studying other response aspects such as spike-time precision or response discrimination. This finding constitutes a huge departure from traditional view, because it arises even when the response aspects are defined through stimulus-independent transformations, seemingly violating the data processing theorems. To that end, we reinterpret the encoding and the decoding problems from the observer perspective, and disentangle them based on their operation, as opposed to the type of measure used to assess them \citep{quiroga2009}. Then, we study the role of spike-time precision and response discrimination in encoding and decoding using stochastic codes, as opposed to reduced codes \citep{victor1996,victor2005,nelken2008,quiroga2009,kayser2010,rusu2014}. We show that, like noise correlations, stochastic codes may yield different information losses when employed within the encoding and the decoding problems, but for none of the causes previously proposed by noise-correlation studies \citep{nirenberg2003,schneidman2003,latham2005,averbeck2006}. Finally, we determine under which conditions the encoding and decoding problems are reliably related, thereby clarifying traditional views. We conclude that response aspects including and beyond noise correlations may play different roles in encoding and decoding, and that decoders designed with noisy data may perform optimally when operating on quality data. In this way, we open up new opportunities for studying neural codes and brain computations that may potentially yield experimental and computational savings and reduce the complexity and cost of neural prosthetics.

\section{Methods}

\subsection{Statistical notation} 

When no risk of ambiguity arises, we here employ the standard abbreviated notation of statistical inference \citep{casella2002}, denoting random variables with letters in upper case, and their values, with the same letters but in lower case. Accordingly, the symbols $ P(x|y) $ and $ P(x) $ always denote the conditional and unconditional probabilities, respectively, of the random variable $ X $ taking the value $ x $ given that the random variable $ Y $ takes the value $ y $. Although this notation is common in neuroscience \citep{nirenberg2003,quiroga2009,latham2013}, we have found that it may potentially lead to confusing the encoding and the decoding problems. In those cases, we employ the more general symbols $ P(X{=}v|Y{=}w) $ and $ P(X{=}v) $ to denote the conditional and unconditional probabilities, respectively, that the random variable $ X $ takes the value $ v $ given that the random variable $ Y $ takes the value $ w $.

\subsection{Encoding} The process of converting stimuli $ \VS $ into neural responses $ \VR $ (e.g., spike-trains, local-field potentials, electroencephalographic or other brain signals, etc.) is called encoding \citep{schneidman2003,panzeri2010}. The information that $ \VR $ contains about $ \VS $ is here quantified as 

\begin{equation}\label{met::eq::infodef}
I(\VS;\VR) = \sum_{\Vs,\Vr}{P(\Vs,\Vr)\log{\frac{P(\Vs|\Vr)}{P(\Vs)}}} \, .
\end{equation}

\noindent More generally, the information $ I(\VS;X) $ that any random variable $ X $, including but not limited to $ \VR $, contains about $ \VS $ can be computed using the above formula with $ \VR $ replaced by $ X $. For compactness, we will denote $ I(\VS;X) $ as $ \infoes{X} $ unless ambiguity arises.

\subsection{Decoding} The process of transforming $ \VR $ into estimated stimuli $ \VST $ (or into perceptions, decisions and actions) is called decoding \citep{schneidman2003,panzeri2010}. Under the Bayesian coding hypothesis \citep{knill2004,bergen2015}, decoding is usually performed using optimal decoders (also called Bayesian or maximum-a-posteriori decoders, and ideal homunculus or observers, among other names; see \cite{eyherabide2016} and references therein). In their most general form, these decoders map $ \VR $ into $ \VST $ according to the following formula

\begin{equation}
\VsT = \arg \max_{\Vs}{P(\VS{=}\Vs|\VRT{=}\Vr)}\, . 
\end{equation}

\noindent The symbol $ \VRT $ denotes the surrogate responses either used to train the optimal decoder, or associated with the approximation $ P(\Vs|\VrT) $ of the real posterior probabilities $ P(\Vs|\Vr) $ used to construct the decoder \citep{merhav1994,nirenberg2003,latham2005,quiroga2009,oizumi2010,eyherabide2013,latham2013,eyherabide2016}. 

The joint probability $ P(\Vs,\VsT) $ of the presented and estimated stimuli, also called confusion matrix \citep{quiroga2009}, can be computed as follows

\begin{equation}\label{met::eq::confusion}
P(\Vs,\VsT){=}\sum_{\Vs,\Vr_{\VsT}}{P(\Vs,\Vr)}\, ,
\end{equation}

\noindent where $ \Vr_{\VsT} $ denotes all $ \Vr $ that are mapped into $ \VsT $. Then, the information that $ \VST $ preserves about $ \VS $ is given by $ \infoes{\VST} $, whereas the decoding accuracy above chance level is here defined as follows

\begin{equation}\label{met::eq::accuracydef}
\Ad{\VR}{\VRT}=\sum_{\Vs}{P(\VS{=}\Vs,\VST{=}\Vs)}-\max_{\Vs}{P(\Vs)}\, ,
\end{equation}

\noindent with the subscript and superscript in $ \Ad{X}{Y} $ always indicating the neural codes used to construct and to operate optimal decoders, in that order. The data processing theorems ensure that both $ \Ad{\VR}{\VR}{\geq}\Ad{\VRT}{\VRT} $ and $ \infoes{\VR}{\geq}\infoes{\VST} $.

\subsection{Neural codes} Any representation of the recorded neural activity into a vector $ \VR{=}[R_1,\ldots,R_N] $ of $ N $ response aspects is usually called a neural code. The aspects may represent, for example, neurons, neural populations, or cortical areas \citep{nirenberg2003,schneidman2003,latham2005,eyherabide2013,eyherabide2016}, but also first-spike latencies, spike counts, and spike-timing variability \citep{furukawa2002,nelken2007,eyherabide2010FCN,eyherabide2016}. To assess what aspects are informative, $ \VR $ is usually transformed into another neural code $ \VRT $, often called reduced code \citep{schneidman2003,eyherabide2010FCN}, through stimulus-independent deterministic functions that typically preserve a limited number of response aspects or a coarser version of the existing ones. As we note here, this assessment can also be conducted using stimulus-independent stochastic functions, here called stochastic codes (see Results). In both cases, the data processing theorem ensures that $ \infoes{\VR}{\geq}\infoes{\VRT} $.

\subsection{Noise correlations} 

Given a neural code $ \VR{=}[R_1,\ldots,R_N] $, the aspects are said noise independent if $ P(\Vr|\Vs){=}\PNIs(\Vr|\Vs) $ for all $ \Vr $ and $ \Vs $, with $\PNIs(\Vr|\Vs){=}\prod_{n=1}^{N}{P(r_n|\Vs)}$. Otherwise, they are said noise correlated. As mentioned in \cite{eyherabide2016}, this definition should not be confused with both those that average over stimuli \citep{gawne1993,pereda2005,womelsdorf2012}, and therefore prone to cancellation effects and to confusing noise correlations with activity correlations \citep{nirenberg2003,schneidman2003,eyherabide2016}; and those that are limited to specific linear or nonlinear types of correlations \citep{pereda2005,cohen2011,latham2013}.

\subsection{Encoding-oriented measures} From an information-theoretical standpoint, the importance of response aspects in neural encoding has previously been quantified in many ways \citep[see ][ and references therein]{eyherabide2010FCN}, out of which we chose the following four measures

\setlength\mywidthalign{\widthof{$ \diae{\VR}{\VRT} $}}
\begin{align}
    \mboxalign{$\dienc{\VR}{\VRT}$} &= \infoes{\VR}-\infoes{\VRT}\label{met::eq::dier}\\
    \mboxalign{$\diae{\VR}{\VRT}$} &= \Ad{\VR}{\VR}-\Ad{\VRT}{\VRT} \label{met::eq::daer}\\
    \mboxalign{$\dienc{\VR}{\VSTL}$} &= \infoes{\VR}-\infoes{\VSTL}\label{met::eq::diels}\\
    \mboxalign{$\dienc{\VR}{\VST}$} &= \infoes{\VR}-\infoes{\VST}\label{met::eq::dies}
\end{align}

\noindent Here, $ \VR $ and $ \VRT $ are two arbitrary neural codes; whereas $ \VSTL$ and $ \VST $ denote a sorted stimulus list and the most-likely stimulus, both according to $ P(\VS{=}\Vs|\VRT{=}\VrT) $. The measure $ \dienc{\VR}{\VSTL} $ is here introduced as an encoding-oriented version of the decoding-oriented measure $ \diq{Q}{LS} $ previously proposed by \cite{ince2010} and defined after the next two sections together with other decoding-oriented measures. Notice that $ \dienc{\VR}{\VSTL}{=}\dienc{\VR}{\VST} $ when the number of stimuli is two.

In this study, stochastic functions always exist that transform $ \VR $ into $ \VRT $ in a stimulus-independent manner, unless otherwise stated. Hence, the transformation from $ \VR $ into $ \VST $ can be interpreted as a sequential process \citep{geisler1989,eyherabide2013} in which $ \VR $ is first transformed into $ \VRT $, then into $ \VSTL $, and finally into $ \VST $. Consequently, $ \dienc{\VR}{\VRT} $, $ \dienc{\VR}{\VSTL} $ and $ \dienc{\VR}{\VST} $ can be interpreted as accumulated information losses after the first, second and third transformations, respectively, whereas $ \diae{\VR}{\VRT} $, as the accuracy loss after the first transformation. The data processing theorems ensure that all the above measures are never negative.

The above four measures are here regarded as encoding-oriented, even though the last three use optimal decoders. Indeed, our classification does not follow previous criteria based on the nature of the measure \citep{schneidman2003,quiroga2009}, which are often obscure and arguably questionable \citep{eyherabide2016}. Specifically, consider $\dienc{\VR}{\VRT} $ and $\dienc{\VR}{\VST} $. The former has always been regarded as an encoding measure even when no stimulus-independent function from $ \VR $ into $ \VRT $ exists \citep{nirenberg2003}, and even though it need not be conclusive about the actual encoding mechanism (see Results). Classifying the latter as an encoding measure may seem questionable because it is based on encoded information, but in the output of a decoder \citep{quiroga2009,eyherabide2016}. However, in the light of previous studies on the role of noise correlations in neural decoding \citep{nirenberg2001,nirenberg2003,latham2005,eyherabide2016}, classifying it as a decoding measure may also seem questionable because the decoder is trained and tested with the same responses.

\subsection{Noise correlations in neural encoding} Notice that \refeq{met::eq::dier} has previously been used to quantify the importance of noise correlations in encoding information \citep{schneidman2003,eyherabide2009}. In those cases, $ \VRT $ has been replaced with surrogate responses generated assuming that neurons are noise independent, here denoted $ \VRNI $. However, we can prove the following theorem:

\begin{theorem}\label{met::theo::corrsto}
$ \VRNI $ may not be related to $ \VR $ through stochastic codes and is never related to $ \VR $ through a reduced code.
\end{theorem} 

\begin{proof} The first part is proved in Results. The second part was first proved in \cite{schneidman2003}, but their proof is invalid when $ \VRNI $ contains the same responses as $ \VR $. In that case, we can prove the second part by contradiction, assuming that a mapping exists from $ \VR $ into $ \VRNI $ which is deterministic, and hence bijective. Therefore, both $ \VRNI $ and $ \VR $ maximize the conditional entropy given $ \VS $ over the probability distributions with the same the marginals. Because the probability distribution achieving this maximum is unique \cite{cover2006}, $ \VRNI $ and $ \VR $ must be the same, thereby proving our statement.\end{proof}

Consequently, $ \dienc{\VR}{\VRNI} $, $ \diae{\VR}{\VRNI} $, $ \dienc{\VR}{\VSNIL} $, and $ \dienc{\VR}{\VSNI} $ (with $ \VSNIL $ and $ \VSNI $ analogous to $ \VSTL $ and $ \VST $, respectively, but for optimal decoders operating on $ \VRNI $) can be ensured neither to be nonnegative nor to compare information of the same type \citep{nirenberg2003}, unless a stochastic code exists that maps $ \VR $ into $ \VRNI $.

\subsection{Decoding-oriented measures} From and information-theoretical standpoint, the importance of response aspects in neural decoding has been previously quantified in many ways \citep[see ][ and references therein]{eyherabide2013}, out of which we chose the following five that are more closely related to optimal decoding

\setlength\mywidthalign{\widthof{$ \diq{Q}{DL} $}}
\begin{align}
\mboxalign{$ \diq{Q}{D}$}  &= \sum_{\Vs,\Vr}{P(\Vs,\Vr)\,\ln{\frac{P(\Vs|\Vr)}{P(\VS{=}\Vs|\VRT{=}\Vr)}}}\label{met::eq::diqd}\\
\mboxalign{$ \diq{Q}{DL}$}&= \min_{\theta}{\sum_{\Vs,\Vr}{P(\Vs,\Vr) \ln{\frac{P(\Vs|\Vr)}{P(\VS{=}\Vs|\VRT{=}\Vr,\theta)}}}} \label{met::eq::diqdl}\\ 
\mboxalign{$ \diq{Q}{LS}$}&= \infoes{\VR}-\infoes{\VSTL} \\
\mboxalign{$ \diq{Q}{B} $} &= \infoes{\VR}-\infoes{\VST} 
\end{align}

\noindent Here, $ \VR $ is the neural code on which the decoders operate; $ \VRT $, the neural code with which the decoders are trained; $ \theta $, a real scalar; $ K $, the number of stimuli; $ \VSTL$ and $ \VST $, a sorted stimulus list and the most-likely stimulus, both according to $ P(\VS{=}\Vs|\VRT{=}\Vr) $; and we define $ P(\Vs|\VrT,\theta)$ below, after resolving the flaws of previous definitions \citep{merhav1994,latham2005,oizumi2009,oizumi2010,eyherabide2013,oizumi2016}. Specifically, we can prove the following theorem

\begin{theorem}\label{met::theo::dinidl1}
$ P(\Vs|\VrT,\theta){\propto}P(\Vs)$ for all $ \Vs $ and $ \VrT $ if a $ \VrB $ exists such that $P(\VR{=}\VrB|\VS{=}\Vs){>}P(\VRT{=}\VrB|\Vs){=}0$ for some or all $ \Vs $. 
\end{theorem}

\begin{proof} Using the standard conventions that $ 0 \log{0} {=} 0 $ and $ x \log{0} {=}\infty $ for $ x{>}0 $ \citep{cover2006}, Eq (B13a) in \cite{latham2005} is fulfilled for all $ P(\Vs|\VrT,\theta) $ such that $ P(\Vs|\VrT,\theta){>}0 $ when $P(\VR{=}\VrT|\VS{=}Vs){>}P(\VrT|\Vs){=}0$. Our result immediately follows by solving Eq  (B15) in \cite{latham2005} with $ \beta{=}0 $. 
\end{proof}

In addition, we can generalize the flaws of previous definitions of $ P(\Vs|\VrT,\theta) $ found in \cite{eyherabide2013} as follows 

\begin{theorem}\label{met::theo::dinidl2}
$ P(\Vs|\VrT,\theta){<}\infty$ if $P(\VrT|\Vs){=}P(\VR{=}\VrT){=}0 $ regardless of $ \Vs $. Otherwise, $ P(\Vs|\VrT,\theta){=}0$ if $P(\VrT|\Vs){=}P(\VR{=}\VrT|\VS{=}\Vs){=}0 $ for some but not all $ \Vs $.
\end{theorem}

\begin{proof} Using the standard convention that $ 0 \log{0} {=} 0 $, Eq (B13a) in \cite{latham2005} is fulfilled only if $ P(\Vs,\VrT|\theta){=}0 $ when $P(\Vs,\VrT){=}0$ and $P(\Vs,\Vr){=}0 $ for $ \Vr{=}\VrT $. Our result immediately follows using Bayes' rule. 
\end{proof}

Consequently, here we define $ P(\Vs|\VrT,\theta) $ as follows

\begin{equation}
P(\Vs|\VrT,\theta)\propto \left\{\begin{array}{ll}
P(\Vs) & \mbox{if $\exists \Vs,\VrB $ such that $P(\VR{=}\VrB|\VS{=}\Vs){>}P(\VRT{=}\VrB|\VR{=}\Vs){=}0$}\\
0 & \mbox{if $P(\VrT|\Vs){=}P(\VR{=}\VrT|\VS{=}\Vs){=}0 $ for some but not all $ \Vs $}\\
P(\Vs)\,P(\VrT|\Vs)^{\theta} & \mbox{otherwise}
\end{array}\right.
\end{equation}

In addition, we will compute the accuracy loss from the decoding perspective as follows

\begin{equation}
\diaq{Q}{B}= \Ad{\VR}{\VR}-\Ad{\VR}{\VRT}\, .
\end{equation}

\noindent This measure differs from $ \diae{\VR}{\VRT} $ in that the last term uses optimal decoders that operate on $ \VR $, as opposed to $ \VRT $. Notice that the decoding-oriented measures defined above have previously been used to quantify the importance of noise correlations in neural decoding, for example, by replacing $ \VRT $ with $ \VRNI $ \citep{nirenberg2001,nirenberg2003,latham2005,eyherabide2009,oizumi2009,quiroga2009,ince2010,oizumi2010,eyherabide2013,latham2013,eyherabide2016}.

\subsection{Shortcomings of decoding-oriented measures} 

The measures $ \diq{Q}{D} $, $ \diq{Q}{DL} $, $ \diq{Q}{LS} $ and $ \diq{Q}{B} $ have all been previously regarded as quantifying the information loss caused when optimal decoders operate on $ \VR $, but make decisions assuming that the input is $ \VRT $. The first measure, $ \diq{Q}{D} $, was introduced by \cite{nirenberg2001} with $ P(\VS{=}\Vs|\VRT{=}\Vr){=}\PNIs(\Vs|\Vr) $, and later extended to any $ P(\VS{=}\Vs|\VRT{=}\Vr) $ \citep{nirenberg2003,quiroga2009,latham2013,eyherabide2016}. Although often overlooked \citep{quiroga2009,latham2013}, $ \diq{Q}{D} $ presents the following two difficulties: first, it can exceed $ \infoes{\VR} $ and tend to infinity if $P(\Vs|\VrT){=}0$ when $P(\Vs|\Vr){>}0 $ for some $ \Vs $ and $ \Vr{=}\VrT $ \citep{schneidman2003}; and second, here we note that it may be undefined if $P(\VrT){=}0$ when $P(\Vr){>}0 $ for some $ \Vr{=}\VrT $ (\refsfig{generalizereducecodes}{B}).

The second measure, $ \diq{Q}{DL} $, was introduced by \cite{latham2005} and has always been deemed exact \citep{latham2005,oizumi2009,quiroga2009,ince2010,oizumi2010,latham2013,oizumi2016}. \cite{latham2005} showed that, unlike $ \diq{Q}{D} $, $ \diq{Q}{DL}{\leq}\infoes{\VR}$, but their proof ignored the cases mentioned in our theorems~\ref{met::theo::dinidl1}~and~\ref{met::theo::dinidl2}. Nevertheless, these theorems imply that $ \diq{Q}{DL}{\leq}\infoes{\VR}$ through \refeq{met::eq::diqdl}. However, \cite{eyherabide2013} has shown that $ \diq{Q}{DL} $ may exceed both $ \diq{Q}{LS} $ and $ \diq{Q}{B} $; and \cite{eyherabide2016}, that it may overestimate the loss when decoding together neural populations that transmit independent information.

The third measure, $ \diq{Q}{LS} $, was introduced by \cite{ince2010} and quantifies the difference between two encoded informations: the one in $ \VR $, and the one in the output of decoders that, after observing $ \Vr $, produce a stimulus list sorted according to $ P(\VS{=}\Vs|\VRT{=}\Vr) $. The fourth measure, $ \diq{Q}{B} $, quantifies the difference between the information encoded in $ \VR $ and that encoded in the output of an optimal decoder constructed using $ P(\Vs|\VrT) $. When the number of stimuli is two, $ \diq{Q}{LS} $ reduces to $ \diq{Q}{B} $ \citep{eyherabide2013,eyherabide2016}.

As we show in \reffig{generalizereducecodes}{B}, $ \diq{Q}{D} $, $ \diq{Q}{LS} $, $ \diq{Q}{B} $, and $ \diaq{Q}{B} $ are undefined if the actual responses $ \VR $ are not contained within the surrogate responses $ \VRT $. The result for $ \diq{Q}{D} $ was already explained in the first paragraph. For the other measures, our observation stems from the fact that, for the missing responses, an optimal decoder simply does not know what output should be produced. Our observation does not apply to $ \diq{Q}{DL} $ due to our theorems~\ref{met::theo::dinidl1}~and~\ref{met::theo::dinidl2}, but can be easily seen from the definition of $ \diq{Q}{DL} $ as a communication information loss \citep{latham2005,oizumi2010,eyherabide2016}. Indeed, the fact that optimal decoders may fail to operate on $ \VR $ due to the missing responses in $ \VRT $ indicates that no information can be reliably transmitted, thereby yielding $ \diq{Q}{DL}{=}\infoes{\VR} $ \citep{latham2005,eyherabide2016}.

Finally, notice that decoding measures may yield different and sometimes contradictory results \citep{oizumi2009,oizumi2010,eyherabide2013,latham2013}. However, the contradictions can be partially resolved by noticing that they have been previously derived using different notions of information. The reader is referred to \cite{{eyherabide2016}} and references therein for further remarks on these measures and notions.

\section{Results and Discussion}

\subsection{Encoding and decoding from the observer perspective}

The role of first-spike latencies, spike counts, or other response aspects in encoding information --- the encoding problem --- and in brain computations --- the decoding problem --- has previously been assessed, for example, using methods based either on information theory or on decoders \citep{quiroga2009}. For decades, methods based on information theory have often been related to the encoding problem, whereas methods based on decoders have often been related to the decoding problem. However, here we argue that the nature of the employed method is not reliably related to its operational significance, at least from the observer perspective.

Information-based methods typically compare the information encoded in two representations of the neural response, usually called neural codes and here denoted $ \VR $ and $ \VRT $. The code $ \VR $ typically preserves all response aspects, whereas the code $ \VRT $ usually arises from some transformation of $ \VR $ that ignores some response aspects, or at least the additional information they carry. The ignored response aspects are usually deemed inessential for encoding information when the information in $ \VR $ and $ \VRT $ is approximately equal, and important otherwise. However, the encoded information  need not provide conclusive insight about the encoding mechanisms, but about the read out mechanisms.

To illustrate this, consider the hypothetical experiment shown in \reffig{effectofbinning}{A}. There, the responses of a single neuron elicited by two visual stimuli have been characterized by counting the number of spikes within consecutive time-bins starting from the stimulus onset.  This characterization preserves all the encoded information regardless of whether the bin size is $ 5 $, $ 10 $, or $ 15\,ms $, thereby shedding limited insight into the actual spike-time precision employed by the encoding mechanism, which is $ 10\,ms $. However, this result certainly proves that such knowledge is inessential for ideal observers, who can read the spike trains without loss using any of the aforementioned bin sizes.

\begin{figure}[h!]
\renewcommand{\baselinestretch}{1.0}
\begin{center}
\includegraphics[]{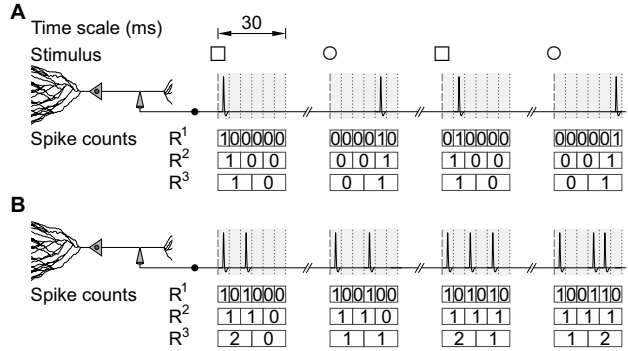}
\end{center}
\caption[]{\textbf{Lossless codes need not disentangle the encoding mechanisms.}\label{effectofbinning} A: Hypothetical intracellular recording of the precise spike patterns elicited by a single neuron after presenting in alternation two visual stimuli, namely $ \protect\SBox $ and $ \protect\SCirc $. Stimulus-response probabilities are arbitrary, and time is discretized in bins of $ 5\,ms $. The responses are recorded within $ 30\,ms $ time-windows after stimulus onset. Each type of frame elicits single spikes fired with different latencies, uniformly distributed between $ 0 $ and $ 10\,ms $ after the onset of $ \protect\SBox $, and between $ 20 $ and $ 30\,ms $ after the onset of $ \protect\SCirc $. Responses have been characterized by counting the number of spikes within consecutive time-bins of size $ 5 $, $ 10 $ and $ 15\,ms $ starting from the stimulus onset, thereby yielding a discrete-time sequences here denoted $ \VR^1 $, $ \VR^2 $ and $ \VR^3 $, respectively. B: Analogous description to panel A, but with each type of frame producing two different types of response patterns composed of $ 2 $ or $ 3 $ spikes. 
}
\end{figure}

Decoder-based methods usually compare the decoded informations or the decoding accuracies of two decoders: one constructed using $ \VR $ and another one constructed using $ \VRT $. The ignored response aspects are usually deemed inessential for decoding information when the information extracted by these two decoders is approximately equal, and important otherwise. However, the information extracted by optimal decoders trained and tested with surrogate responses $ \VRNI $ generated assuming that neurons are noise independent \citep{quiroga2009} need not provide insight into the necessity for taking into account noise correlations when constructing optimal decoders \citep{nirenberg2003,eyherabide2016}.

Our observations motivate us to reformulate the coding problem from the perspective of an ideal observer or organism (\refsfig{encodingdecoding}{A}). Within this perspective, we devise the encoding problem as concerning the losses caused when the ideal observers either can only see and are only allowed to see not the actual neural responses ($ \VR $), but a transformed version of them ($ \VRT $) that preserves a limited number of response aspects (\refsfig{encodingdecoding}{C}). The decoding problem is here interpreted analogously to the decoding perspective of the role of noise correlations, namely, as concerning the losses caused when observers make decisions assuming that they see $ \VRT $ instead of $ \VR $ (\refsfig{encodingdecoding}{D}). Mathematically, both problems can be tackled by training ideal observers not with $ \VR $ but with $ \VRT $ (\refsfig{encodingdecoding}{B}). However, within the encoding problem, the ideal observer operates on $ \VRT $, whereas within the decoding problem, the ideal observer operates on $ \VR $.

\begin{figure}[h!]
\renewcommand{\baselinestretch}{1.0}
\hfill
\begin{center}
\includegraphics[]{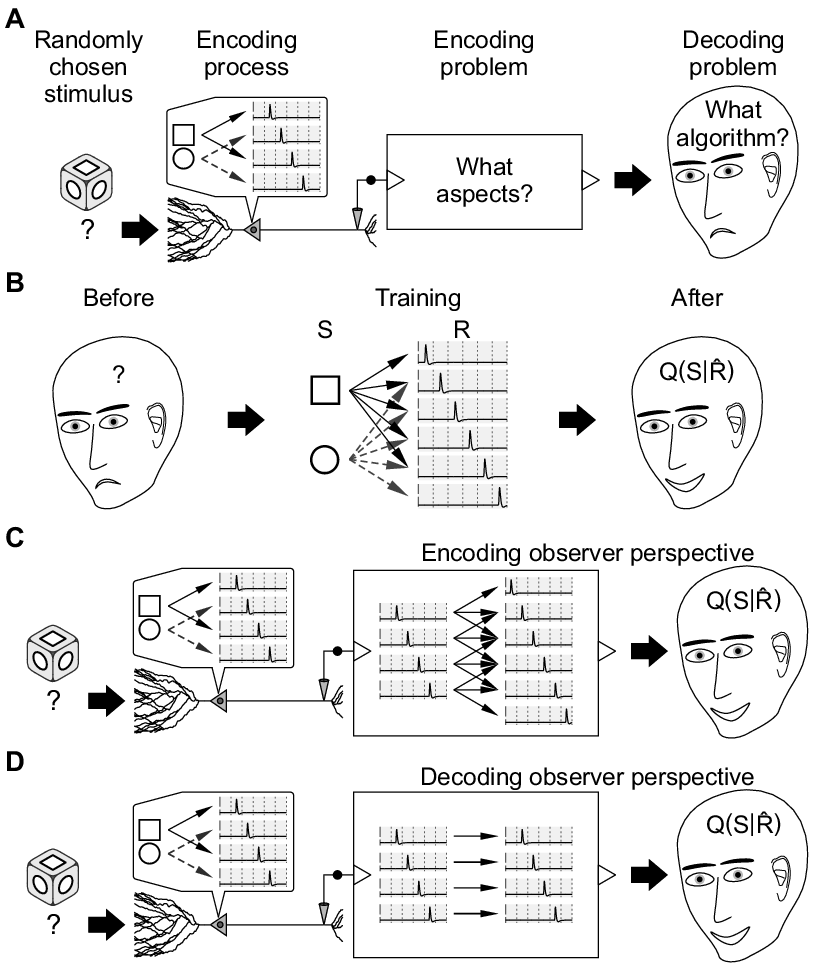}
\end{center}
\caption{\textbf{Encoding and decoding problems from the organism perspective.}\label{encodingdecoding}  A: The encoding problem is here reinterpreted as concerning the losses caused by transforming the recorded neural responses $ \VR $ into another responses $ \VRT $ before supplying them to ideal observers \citep{nelken2007}; whereas the decoding problem, as concerning the losses caused by training ideal observers with transformed responses $ \VRT $ or approximate response likelihoods, as opposed to the actual responses that the observer sees \citep{nirenberg2003,latham2005,quiroga2009,latham2013}. B: Both problems can be tackled by training ideal observers with the transformed responses $ \VRT $. C: However, in the encoding problem, the above observer is provided with the neural responses after the transformation, namely $ \VRT $. D: On the contrary, in the decoding problem, the above observer is provided with the original neural responses unchanged, namely $ \VR $. }
\end{figure}

\subsection{Relation between encoding and decoding}

These conceptual and operational differences notwithstanding, decoding-oriented measures of loss have previously been regarded as generalizing their encoding-oriented counterparts \citep{nirenberg2001,nirenberg2003,quiroga2009,latham2013}. Unfortunately, these claims ought to be observed with caution for at least three reasons: they need not be accurate, potentially confuse encoding and decoding, and are limited to specific types of codes. In this section, we will explain and justify these reasons.

Concerning the accuracy of the claims, \cite{nirenberg2003} claimed that their version $ \diq{Q}{D} $ for studying noise correlations is always positive and reduces to the difference in encoded information between two neural codes: one using large time-bins and another one using small time-bins. However, this difference can be positive or negative, regardless of the order in which the informations are subtracted, thereby contradicting the putative positivity of both $ \diq{Q}{D} $ and the measure they proposed. Indeed, in \refsfig{effectofbinning}{B}, using time-bins of $ 5$ or  $15\,ms $ preserves all the encoded information, but using intermediate time-bins of $ 10\,ms $ preserves none, thereby proving our statement.

Our finding seemingly contradicts previous experimental studies \citep{reinagel2000,quiroga2013} and, most importantly, the data processing theorem \citep{cover2006}. Indeed, it shows that previous experimental findings of encoded information monotonically decreasing with bin size  are not completely attributable to time discretization. However, it does not violate the data processing theorem because the code with $ 15\,ms $ time-bins cannot be derived from that with $ 10\,ms $ time-bins through stimulus-independent transformations.

Actually, the study of \cite{nirenberg2003} proposed a generalization of $ \diq{Q}{D} $, as opposed to a reduction of it, but unfortunately it cannot be interpreted in the same way. Indeed, their generalization compares the additional costs of using two different codes for constructing two decoders, respectively, that are then employed to decode a third code. Thus, this generalization need not be suitable for addressing the decoding problem, which concerns the losses caused by decoding a code assuming that it is another one, unless it is first reduced to $ \diq{Q}{D} $.

The confusion between encoding and decoding arises from the particular choice of probabilities that previous studies \citep{nirenberg2003,quiroga2009,latham2013} have employed when computing $ \diq{Q}{D} $ for stimulus-independent transformations of the original codes representing the neural responses, but we defer the proofs until the last section of Results. Most importantly, the above generalization and computations of $ \diq{Q}{D} $, as conducted in previous studies \citep{nirenberg2003,quiroga2009,latham2013}, together with their criticism \citep{schneidman2003}, are actually limited to stimulus-independent deterministic transformations, often called reduced codes, as opposed stochastic ones, here called stochastic codes. In the next section we characterize this codes, and in the one immediately after, we will assess what role they play in encoding and decoding from the observer perspective.

\subsection{Stochastic codes generalize reduced codes}

Stochastic codes arise naturally when studying the importance of specific response aspects \citep{eyherabide2010FCN}. However, stochastic codes are most valuable when studying other response aspects such as spike-time precision \citep{kayser2010,quiroga2009}, response discrimination \citep{victor1996,victor2005,rusu2014}, noise along neural pathways \citep{nelken2008}, and even noise correlations. In these cases, the use of reduced codes may be questionable, if not impossible. In this section, we will first illustrate these claims before giving a formal definition of stochastic codes.

Consider the hypothetical experiment in \reffig{generalizereducecodes}{A}, in which the neural responses $ \VR{=}[L,C] $ can be completely characterized by the first-spike latencies ($ L $) and the spike counts ($ C $). The importance of $ C $ can be studied not only using a reduced code that replaces all its values with a constant (\refsfig{generalizereducecodes}{B}), but also a stochastic code that shuffles the values of $ C $ across all responses with the same $ L $ (\refsfig{generalizereducecodes}{C}), or even a stochastic code that preserves the original value of $ L $ but chooses the value of $ C $ from some possibly $L-$dependent arbitrary probability distribution (\refsfig{generalizereducecodes}{D}).

\begin{figure}[h!]
\renewcommand{\baselinestretch}{1.0}
\hfill
\begin{center}
\includegraphics[]{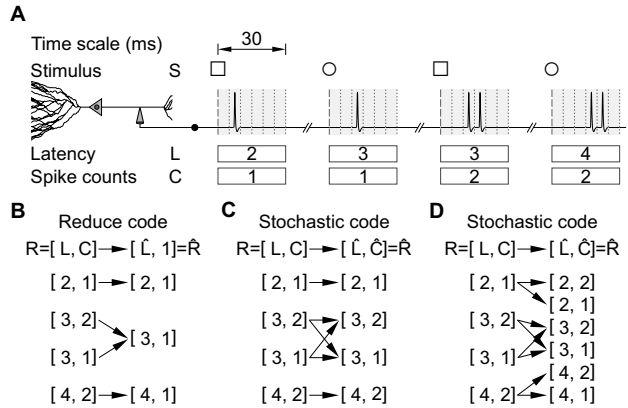}
\end{center}
\caption[]{\textbf{Stochastic codes include and generalize reduced codes.}\label{generalizereducecodes} A: Analogous description to \reffig{effectofbinning}{A}, but with responses characterized using a code $ \VR{=}[L,C] $ that preserves the first-spike latency ($ L $) and the spike-count ($ C $). B: Deterministic transformation (arrows) of $ \VR $ in panel A into a reduced code $ \VRT{=}[L,1] $, which ignores the additional information carried in $ C $ by considering it constant and equal to unity. This reduced code can also be reinterpreted as a stochastic code with transition probabilities $ Q(\VRT|\VR) $ defined by \refeq{app::eq::Qfig2B}. C: The additional information carried in $ C $ is here ignored by shuffling the values of $ C $ across all trails with the same $ L $, thereby turning $ \VR$ in panel A into a stochastic code $ \VRT{=}[\hat{L},\hat{C}] $ with transition probabilities $ Q(\VRT|\VR) $ defined by \refeq{app::eq::Qfig2C}. D: The additional information carried in $ C $ is here ignored by replacing the actual value of $ C $ for one chosen with some possibly $ L $-dependent probability distribution (\refeq{app::eq::Qfig2D}).}
\end{figure}

Using either reduced or stochastic codes in the previous example seemingly yields the same result (but see next section). However, this need not always be the case. For example, spike-time precision has previously been studied using at least three different methods, here called binning, bin-shuffling and jittering. Binning is perhaps the most common of the three methods and consists in transforming the recorded spike trains into discrete-time sequences (\reffig{effectofbinning}{}), thereby yielding a reduced code that can be interpreted as a stochastic code (\reffig{generalizereducecodes}{}). Bin-shuffling consists in randomly permuting spikes and silences within consecutive non-overlapping windows (\refsfig{naturalstochasticcodes}{A}), which inherently constitutes a stochastic code.

\begin{figure}[h!]
\renewcommand{\baselinestretch}{1.0}
\begin{center}
\includegraphics[]{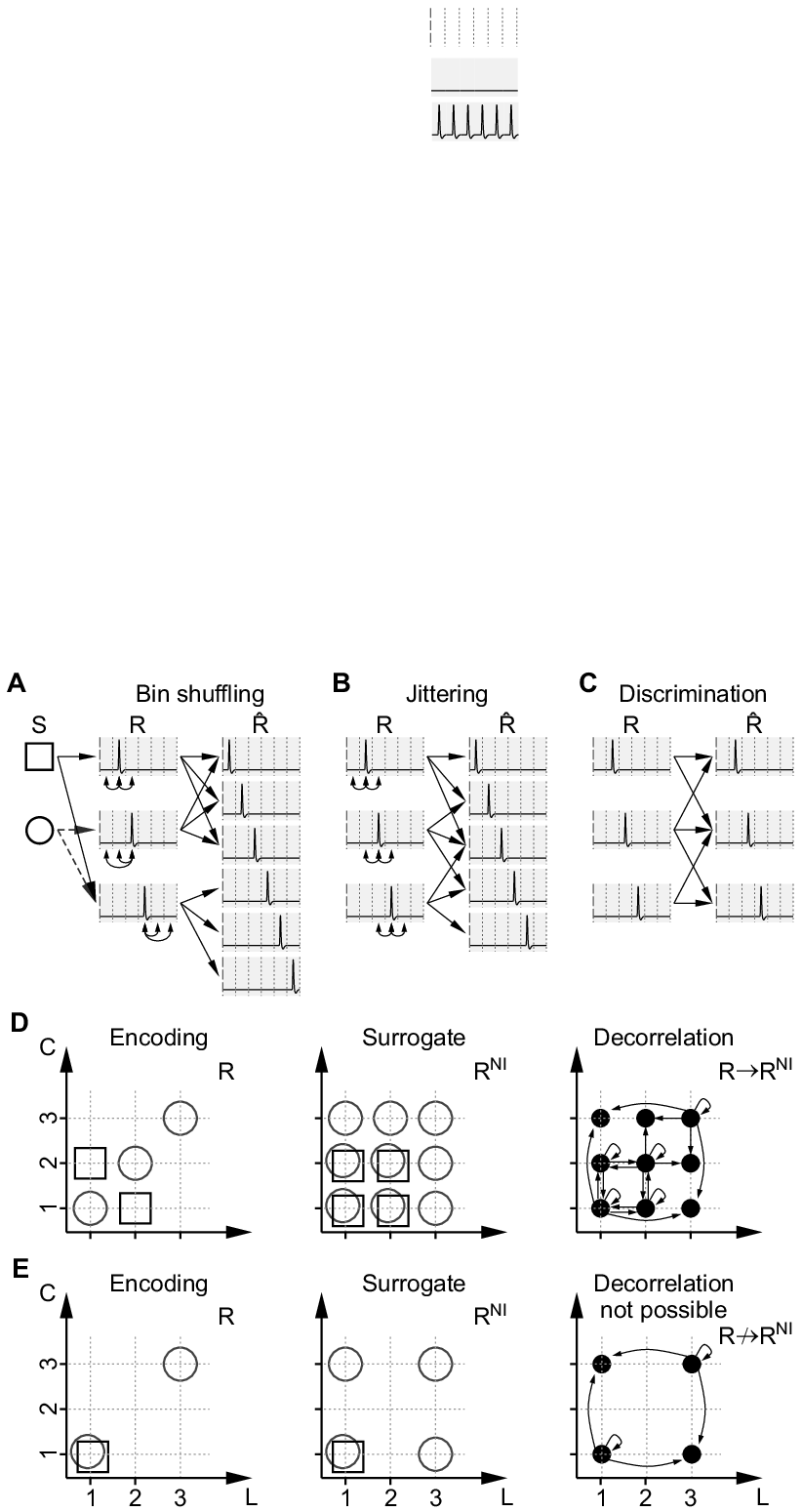}
\end{center}
\caption[]{\textbf{Examples of stochastic codes.} Hypothetical experiments analogous to that described in \reffig{generalizereducecodes}{A}.  
A: Stochastic function (arrows on the left) modeling the encoding process, followed by another one with transition probabilities $ Q(\VRT|\VR) $ given by \refeq{app::eq::Qfig7A}. This function turns $ \VR $ into a stochastic code $ \VRT$ by shuffling spikes and silences within consecutive bins of $ 15\,ms $ starting from stimulus onset. Notice that $ \VRT $ is analogous to $ \VR^3 $ in \refsfig{effectofbinning}{} (see Results). 
B: Responses $ \VR $ in panel A, followed by a stochastic function with $ Q(\VRT|\VR) $ given by \refeq{app::eq::Qfig7B}, which introduces jitter uniformly distributed within $ 15\,ms $ windows centered at each spike. 
C: Responses $ \VR $ in panel A, followed by a stochastic function with $ Q(\VRT|\VR) $ given by \refeq{app::eq::Qfig7C}, which models the inability to distinguish responses with spikes occurring in adjacent bins, or equivalently, with distances $D^{spike}[q{=}1]{\leq}1$ or $D^{interval}[q{=}1]{\leq}1$ (see \cite{victor1996} and \cite{victor2005} for further remarks on these distances). Notice that $ \VRT $ contains the same responses as $ \VR $.
D: Cartesian coordinates depicting: on the left, responses $ \VR $ of a neuron different from that in panel A, with $ L $  and $ C $ are positively correlated when elicited by $ \protect\SCirc $, and negatively correlated when elicited by $ \protect\SBox $; in the middle, the surrogate responses $ \protect\VRNI $ that would occur should $ L $ and $ C $ be noise independent (middle); and on the right, a stimulus-independent stochastic function that turns $ \VR $ into $ \VRNI $ with $ Q(\VRNI|\VR) $ given by \refeq{app::eq::Qfig8}.
E: Same description as in panel D, but with $ L $  and $ C $ noise independent given $ \SBox $, and with the stochastic function depicted on the right turning $ \VR $ into $ \VRNI $ given $ \SCirc $ but not $ \SBox $.}\label{naturalstochasticcodes}
\end{figure}

These two methods are equivalent in the sense that the responses generated by each method can be related through stimulus-independent stochastic functions to those generated by the other, and therefore contain the same amount and type of information. However, both methods suffer from the same drawback: they both treat spikes differently depending on their location within the time window. Indeed, both methods preserve the distinction between two spikes located at different time bins, but not within the same time bin, even if the separation between the spikes is the same.

The third method, jittering, consists in shuffling the recorded spikes within time windows centered at each spike (\refsfig{naturalstochasticcodes}{B}). Notice that the responses generated by this method need not be related to the responses generated by the other two methods through stimulus-independent stochastic functions, nor vice versa.  However, this method inherently yields a stochastic code, and, unlike the previous methods, treats all spikes in the same manner.

As another example, consider the effect of response discrimination, as studied in the seminal work of \cite{victor1996}. There, two responses were considered indistinguishable when some measure of distance between the responses was less than some predefined value. However, the distances were there used for transforming neural responses through a method based on cross-validation that, as we note here, is not guaranteed to be stimulus-independent as stochastic codes. Other methods exist which simply merge indistinguishable responses thereby yielding a reduced code, but these methods are limited to distances that are transitive.

Here we overcome these limitations by devising a method based on stochastic codes, which mimics the operation of an observer that confuses responses by randomly treating them as if it was another response. This method includes those based on reduced codes, and even those based on possibly non-symmetrical and non-metric measures of response similarity. Unlike the previous methods, this one can effectively be applied to the example shown in \refsfig{naturalstochasticcodes}{C}.

Formally, Stochastic codes are here defined as neural codes $ \VRT $ that can be obtained through stimulus-independent stochastic functions of other codes $ \VR $. After observing that $ \VR $ adopted the value $ \Vr $, these functions produce a single value $ \VrT $ for $ \VRT $ chosen with transition probabilities $ Q(\VrT|\Vr) $ such that 

\begin{equation}\label{met::eq::stochasticcondition}
P(\VrT|\Vs) = \sum_{\Vr}{P(\Vr|\Vs)\,Q(\VrT|\Vr)}\, . 
\end{equation}

\noindent When for each $ \Vr $ exists a $ \VrT $ such that $ Q(\VrT|\Vr){=}1 $, stochastic codes become reduced codes (\refsfig{generalizereducecodes}{}).

The problem of finding a feasible $ Q(\VrT|\Vr) $ can be readily solved using linear programming \citep{boyd2004}, which we can simplify by noticing the following two properties

\begin{itemize}
\item[1)] $ Q(\VrT|\Vr){=}0 $ if $ P(\Vs,\Vr){>}0 $ and $ P(\Vs,\VrT){=}0 $, for otherwise it would contradict the definition of stochastic codes; and 
\item[2)] A feasible $ Q(\VrT|\Vr) $ never exists if $ \infoes{\VRT}{>}\infoes{\VR} $, for otherwise it would contradict the data processing theorem.
\end{itemize}    

\noindent For example, in \reffig{naturalstochasticcodes}{D}, we can decorrelate first-spike latencies ($ L $) and spike counts ($ C $) by modeling $ \VRNI $ as a stochastic code derived from $ \VR $. To that end, we can solve \refeq{met::eq::stochasticcondition} for $ Q(\VRT|\VR) $ but with $ \VRT $ replaced by $ \VRNI $. In principle, the condition $ \infoes{\VR}{>}\infoes{\VRNI} $ always hold, and thus a feasible $ Q(\VRNI|\VR) $ may exist (property 2). Such probability must be zero at least whenever $\VRNI{\in}\{[1,3];[2,3];[3,3];[3,2];[3,1]\} $ and $ \VR{\in}\{[1,2];[2,1]\} $ (property 1). One possible solution is given by \refeq{app::eq::Qfig9}.

However, notice that stochastic codes need not always exist. For example, in \reffig{naturalstochasticcodes}{E}, the condition $ \infoes{\VR}{<}\infoes{\VRNI} $ always hold, and therefore no stochastic code can map $ \VR $ into $ \VRNI $. Notice that \cite{schneidman2003} employed an analogous example, but involving different neurons instead of response aspects, only for comparing $ \diq{}{D} $ with $ \dienc{\VR}{\VRT} $,  and ignoring whether $ \VRNI $ constitutes as a stochastic code. Finally, stochastic codes must not be confused with probabilistic population codes \citep{knill2004}, stimulus-dependent noise, decoding algorithms based on cross-validation \citep{victor1996,quiroga2009}, and list decoding \citep{schneidman2003,ince2010,eyherabide2013}.

\subsection{The role of stochastic codes}

Like reduced codes, stochastic codes are here defined as stimulus-independent transformations of another code. Accordingly, both types of codes preserve only information contained in the original code that limits the decoded information when decoders operate on them. Although stochastic codes may include more responses than the original codes, this should arguably raise no concerns \citep{nirenberg2003,schneidman2003} because the transformations that define them are stimulus independent and can be implemented during decoding. Thus, for stochastic codes, the encoding and the decoding problems may be related as they are for reduced codes, even when arising from ignoring noise correlations (\refsfig{naturalstochasticcodes}{D}).

We tested this hypothesis by comparing three encoding-oriented measures of information loss ($ \dienc{\VR}{\VRT} $, $ \dienc{\VR}{\VSTL} $, $ \dienc{\VR}{\VST} $) and one of accuracy loss ($ \diae{\VR}{\VRT} $), with four decoding-oriented measures of information loss ($ \diq{}{D} $, $ \diq{}{DL} $, $ \diq{}{LS} $, $ \diq{Q}{B}$) and one of accuracy loss ($ \diaq{}{B} $). Their definitions, rationale, and classification are given in Methods. We found that, for the stochastic codes in Figures~\ref{generalizereducecodes}--\ref{circularshift}, the encoding-oriented measures were greater or less than the decoding-oriented measures depending on the case and the probabilities (Table~\ref{table1}). Consequently, our results refute the above hypothesis and, most importantly, prove that previous controversies about encoding and decoding transcend the study of noise correlations.

\begin{figure}[h!]
\renewcommand{\baselinestretch}{1.0}
\begin{center}
\includegraphics[]{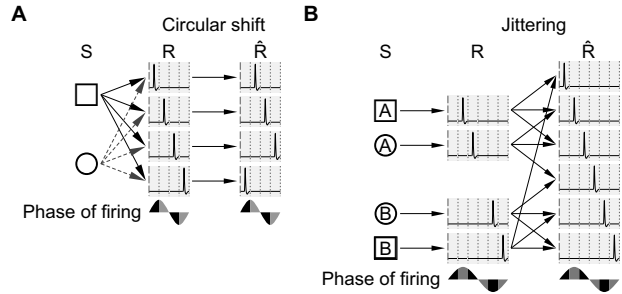}
\end{center}
\caption[]{\textbf{Stochastic codes may play different roles in encoding and decoding.} A: Hypothetical experiment with two visual stimuli, namely $ \protect\SBox $ and $ \protect\SCirc $. The stimuli are transformed (left arrows) into neural responses containing a single spike ($ C{=}1 $) fired at different phases ($ \varPhi $) with respect to a cycle of $ 20\,ms $ period starting at stimulus onset. The phases have been discretized in intervals of size $\pi/2 $ and wrapped to the interval $ [0,2 \pi) $; the resulting value is also known as principal value. The encoding process is followed by a circular phase-shift that transforms $ \VR{=}\varPhi $ into another code $ \VRT{=}\hat{\varPhi} $ with transition probabilities $ Q(\VRT|\VR) $ defined by \refeq{app::eq::Qfig6A}. Notice that $ \VRT $ contains all and only responses present in $ \VR $. B: Analogous description to panel A, except that stimuli are four ($ \SABox $, $ \SACirc $, $ \SBBox $, and $ \SBCirc $), whereas phases are measured with respect to a cycle of $ 30	\,ms $ period and discretized in intervals of size $\pi/3 $. The encoding process is followed by a stochastic transformation (right arrows) that introduces jitter, thereby transforming $ \VR{=}\varPhi $ into another code $ \VRT{=}\hat{\varPhi} $ with transition probabilities $ Q(\VRT|\VR) $ defined by \refeq{app::eq::Qfig6B}. }\label{circularshift}
\end{figure}

\begin{table}[h!]
\renewcommand{\baselinestretch}{1.0}
\renewcommand{\arraystretch}{1.0}
\caption{\textbf{The encoding and decoding problems are unrelated for response aspects beyond noise correlations.} Maximum and minimum differences between encoding-oriented measures of information and accuracy losses and their decoding-oriented counterparts. For each example, the values were computed through the function \texttt{fminsearch} of Matlab 2016, with random initial values for the stimulus-response probabilities and the transition matrices. The computation was repeated until 20 consecutive runs failed to improve the estimate. The values are expressed as percentages of $ \infoes{\VR} $ (the information encoded in $ \VR $) or $ \Ad{\VR}{\VR} $ (the maximum accuracy above chance level when decoders operate on $ \VR $). Notice that $ \dienc{\VR}{\VSTL}{=}\dienc{\VR}{\VST} $ and $ \diq{Q}{LS}{=}\diq{Q}{B} $ because all the cases comprise only two stimuli, and that the absolute value of $\diae{\VR}{\VRT}{-} \diaq{Q}{B} $ can become extremely large when $ \Ad{\VR}{\VR}{\approx}0 $.}\label{table1} 
\begin{center}
\begin{tabular}{@{}|l|r@{\hspace{.5ex}}||r@{\hspace{.5ex}}||@{\hspace{.5ex}}r@{\hspace{.5ex}}|@{\hspace{.5ex}}r@{\hspace{.5ex}}||@{\hspace{.5ex}}r@{\hspace{.5ex}}|@{\hspace{.5ex}}r@{\hspace{.5ex}}||@{\hspace{.5ex}}r@{\hspace{.5ex}}|@{\hspace{.5ex}}r@{\hspace{.5ex}}|@{}}
\hline
\multicolumn{2}{|l||}{Cases} & \refsfig{naturalstochasticcodes}{D} & \refsfig{naturalstochasticcodes}{B} & \refsfig{naturalstochasticcodes}{C} & \refsfig{generalizereducecodes}{D} & \refsfig{naturalstochasticcodes}{A} & \refsfig{generalizereducecodes}{B} & \refsfig{circularshift}{A} \\\hline
\multirow{2}{*}{$\dienc{\VR}{\VRT}{-}\diq{Q}{D}$}  & min & -79& -51& -40 & 0 & 0 & --- & -1193\\ \cline{2-9}
													& max &  26&  32& 46 & 0 & 0 & --- & -171\\ \hline
\multirow{2}{*}{$\dienc{\VR}{\VRT}{-}\diq{Q}{DL}$} & min & -34& -32& -18 & 0 & 0 & -100 & -100\\ \cline{2-9}
													& max &  59&  41& 77 & 0 & 0 & 0 & -100\\ \hline
\multirow{2}{*}{$\dienc{\VR}{\VRT}{-}\diq{Q}{B}$}  & min & -67& -62& -57 & -63 & -87 & --- & -100\\ \cline{2-9}
													& max &  57&  81& 89 & 0 & 0 & --- & -4\\ \hline
\multirow{2}{*}{$\dienc{\VR}{\VST}{-}\diq{Q}{D}$} 	& min & -79& -48& -38 & 0 & 0 & --- & -1192\\ \cline{2-9}
													& max &  74&  92& 92 & 63 & 87 & --- & -71\\ \hline
\multirow{2}{*}{$\dienc{\VR}{\VST}{-}\diq{Q}{DL}$} & min & -34& -27& -18 & 0 & 0 & -100 & -89\\ \cline{2-9}
													& max &  91&  92& 92 & 63 & 87 & 0 & 83\\ \hline
\multirow{2}{*}{$\dienc{\VR}{\VST}{-}\diq{Q}{B}$}  & min & -51& -31& -18& 0 & 0 & --- & -95\\ \cline{2-9}
													& max &  59&  91& 93 & 0 & 0 & --- & 79\\ \hline
\multirow{2}{*}{$\diae{\VR}{\VRT}{-}\diaq{Q}{B}$}   & min &-384&-200& -167 & 0 & 0 & --- & $ \infty $\\ \cline{2-9}
													& max &  95&  67& 100 & 0 & 0 & --- & 0\\ \hline
\end{tabular}
\end{center}
\end{table}

Our conclusions may seem puzzling because the data processing theorems ensure that neither $ \dienc{\VR}{\VRT} $ can exceed the information loss, nor $ \diae{\VR}{\VRT} $ the accuracy loss, caused when decoders operate on $ \VRT $. However, the data processing theorems are not violated because encoding-oriented measures are related to decoders that operate on $ \VRT $, whereas decoding-oriented measures are related to decoders that operate directly on $ \VR $ (\refsfig{encodingdecoding}{D}).

To gain additional insight, we divided the stochastic codes in four groups. In the first group (\refsfig{naturalstochasticcodes}{D}), noise correlations can be ignored through stochastic codes, but the results remain unchanged. Consequently, contrary to previously thought \citep{nirenberg2003}, the discrepancy between encoding- and decoding-oriented measures cannot be completely attributed to surrogate responses generated through stimulus-dependent transformations, or containing responses with information not present in the original neural responses. However, we can prove the following theorem

\begin{theorem}\label{res::theo::nidependence} The transition probabilities $ Q(\VRNI|\VR) $ of stochastic codes that ignore noise correlations may depend both on the marginal likelihoods, and on the noise correlations.
\end{theorem}

\begin{proof}
We can prove the dependency on the marginal likelihoods by computing $ Q(\VRNI|\VR) $ for the hypothetical experiment of \reffig{naturalstochasticcodes}{D}. After some straightforward but tedious algebra, \refeq{met::eq::stochasticcondition} yields, for the $ \VRNI $ associated with $ \SBox $, the following equation must hold
\begin{equation}
P(L{=}1|\VS{=}\SBox) = 0.5\,\left[\delta q + \left(\delta q^2+4\,Q([1,2]|[2,1])\right)^{0.5}\right]\, ,
\end{equation}

\noindent where $ \delta q{=}Q([1,2]|[1,2])-Q([1,2]|[2,1]) $. Hence, any change in $ P(L{=}1|\VS{=}\SBox) $ must be followed by some change in $ Q(\VRNI|\VR) $, thereby proving that the latter cannot be independent of the former.

We can prove the dependency on the noise correlations by computing $ Q(\VRNI|\VR) $ for the hypothetical experiment of \reffig{correlationdependence}{}. After some straightforward but tedious algebra analogous to that in \reffig{naturalstochasticcodes}{D}, \refeq{met::eq::stochasticcondition} yields that, for the responses associated with $ \SBox $, $ Q(\VRNI|\VR) $  must have the following form
\begin{equation}\arraycolsep=1ex
Q(\VRNI|\VR)=\,\left[\begin{array}{cccc}
x_1+x_4 & x_2+x_5 & x_3+x_6 & 1-\sum_{i{=}1}^6{x_i}\\
x_1 & x_2& x_3& 1-\sum_{i{=}1}^3{x_i}\\
x_4 & x_5 & x_6 & 1-\sum_{i{=}4}^6{x_i}\\
0 & 0 & 0 & 1
\end{array}\right]\, ,\end{equation}
or else it will depend on the noise correlations. Here, $ \VR{=}[l,c] $ is associated with the row of index $ i{=}2\,l{-}2{+}c $; whereas $ \protect\VRNI{=}[\hat{l},\hat{c}] $, with the column of index $ j{=}2\,\hat{l}{-}2{+}\hat{c} $; and $ x_1,\ldots,x_6 $ are probabilities. To resolve for $ x_1,\ldots,x_6 $, consider the set of all response distributions with the same marginals as $ \VRNI $ that can be turned into $ \VRNI $ through stimulus-independent stochastic functions regardless of their noise correlations. This set includes $ \VRNI $, and therefore, $ Q(\VRNI|\VR) $ should be able to transform $ \VRNI $ into itself. However, it turns out that this is only possible when $ P(\VRNI|\SBox) $ is unity for $ \VRNI{=}[2,2] $ and zero otherwise, thereby proving that $ Q(\VRNI|\VR) $ depends on the noise correlations in $ \VR $, even if the marginal probabilities are fixed.
\end{proof}

\begin{figure}[h!]
\renewcommand{\baselinestretch}{1.0}
\begin{center}
\includegraphics[]{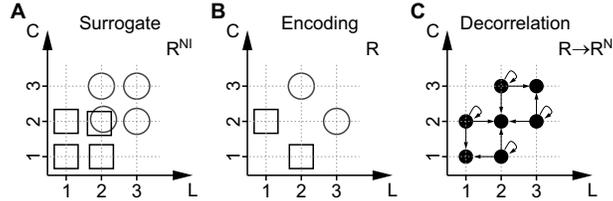}
\end{center}
\caption[]{\textbf{Stochastic codes that ignore noise correlations may depend on them.}\label{correlationdependence} A: Cartesian coordinates representing a hypothetical experiment in which two different stimuli, namely $ \SBox $ and $ \SCirc $, elicit single neuron responses ($ \VRNI $) that completely characterized by their first-spike latency ($ L $) and spike counts ($ C $). Both $ L $  and $ C $ are noise independent. B: Cartesian coordinates representing a hypothetical experiment with the same marginal probabilities $ P(l|\Vs) $ and $ P(c|\Vs) $ as in A: depicting one among many possible types of noise correlations between $ L $ and $ C $. C: Stimulus-independent stochastic function transforming the noise-correlated responses $ \VR $ of \textbf{b} into the noise-independent responses $ \VRNI $ of \textbf{a}. The corresponding transition probability $ Q(\VRNI|\VR) $ is given by \refeq{app::eq::Qfig9}. Even when $ P(l|\Vs) $ and $ P(c|\Vs) $ are fixed, the transition probabilities may depend on the noise correlations, as proved in Results.}
\end{figure}

We tested if the discrepancy between encoding- and decoding-oriented measures stems from the aforementioned dependencies by comparing the values they attain within the second group of stochastic codes (\refsfig{naturalstochasticcodes}{ B,C}). For these codes, which arise naturally when studying spike-time precision or response discrimination, we found that the dependencies never occur (namely, $ Q(\VRT|\VR) $ is independent on $ P(\VS,\VR) $) but the results remain unchanged. Therefore, we conclude that the aforementioned discrepancy cannot be attributed to these dependencies, and can arise not only when studying noise correlations, but when studying other response aspects even if defined through constant and stimulus-response-independent transformations.

Further inspection shows that the discrepancy occurs even when stochastic codes only contain responses present in the original codes (\refsfig{naturalstochasticcodes}{C}), and both in the amount and in the type of information (\refsfig{circularshift}{B}). To prove the latter, consider that an ideal observer is trained using the noisy data $ \VRT $ shown in \reffig{differentinformationtype}{A}, but it is asked to operate on the quality data $ \VR $ shown in \reffig{differentinformationtype}{B}. The information losses $\dienc{\VR}{\VRT}$, $\diq{Q}{D}$, and $\diq{Q}{DL}$ produced by this ideal observer are all equal to $\infoes{\VR}/2 $ regardless of whether $ \VR $ is transformed into $ \VRT $ before showing them to the ideal observer (the encoding problem) or not (the decoding problem).

\begin{figure}[h!]
\renewcommand{\baselinestretch}{1.0}
\begin{center}
\includegraphics[]{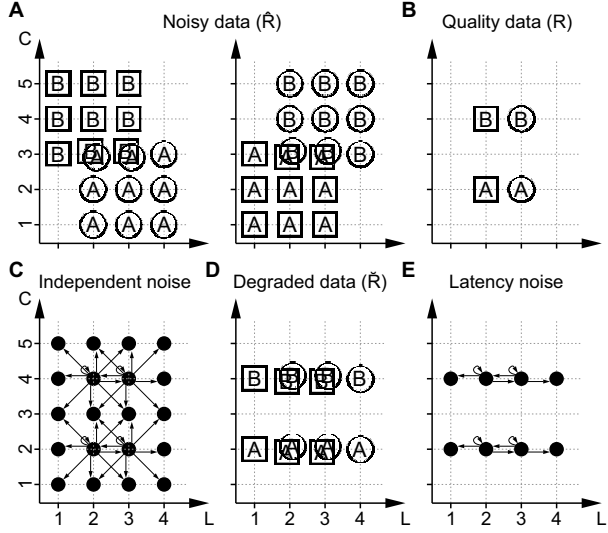}
\end{center}
\caption[]{\textbf{Improving mismatched decoding by adding noise.}\label{differentinformationtype} A: Noisy data recorded in a hypothetical experiment with four compound stimuli $ \VS{=}[\VS_F,\VS_L] $, generated by choosing independently a frame ($ \VS_F{=}\SBox $ or $ \SCirc$) and a letter ($ \VS_L{=}\mathrm{A}$ or $\mathrm{B}$), thereby yielding $ \SABox $, $ \SACirc $, $ \SBBox $, and $ \SBCirc $. Stimuli are transformed (left arrows) into neural responses $ \VRT{=}[\hat{L},\hat{C}] $ with different number of spikes ($ 1{\leq}C{\leq}4 $) fired at different first-spike latencies ($1{\leq}L{\leq}5$; time has been discretized in $ 5\,ms $ bins). Latencies are only sensitive to frames whereas spikes counts are only sensitive to letters, thereby constituting independent-information streams \citep[namely $ P(\Vs,\Vr){=}P(\Vs_F,l)\,P(\Vs_L,c) $; ][]{eyherabide2016}. B: Quality data ($ \VR{=}[L,C] $) recorded as in panel A, but without noise. C: Stimulus-independent stochastic transformation with transition probabilities $ Q(\VRT|\VR) $ given by \refeq{app::eq::Qfig10C}, that introduces independent noise both in the latencies and in the spike counts, thereby transforming $ \VR $ into $ \VRT $ and rendering $ \VRT $ as a stochastic code. D: Degraded data obtained by adding latency noise to the quality data. E: Representation of the stimulus-independent stochastic transformation with transition probabilities $ Q(\VRB|\VR) $ given by \refeq{app::eq::Qfig10D} that adds latency noise in panel D.}
\end{figure}

However, the aforementioned transformation causes some responses to occur for all stimuli, thereby preserving only some information about frames ($ I(\VS_F,\VRT){=}I(\VS_F,\VR)/3 $) and letters ($ I(\VS_L,\VRT){=}2\,I(\VS_L,\VR)/3 $). In other words, within the encoding problem, $ \VRT $ causes a partial information loss $\dienc{\VR}{\VRT}$ that is composed of both frame and letter information. On the contrary, within the decoding problem, $ \VRT $ causes no information loss about letters but a total information loss about frames. We can prove the former by noticing that, for the responses that actually occur in $ \VR $, the ideal observer trained with $ \VRT $ can perfectly identify the letters (because $ P(\hat{C}{=}2|\VS_L{=}\mathrm{A}){=}P(\hat{C}{=}4|\VS_L{=}\mathrm{B}){=}1$). We can prove the latter by noticing that $ P(\hat{l}|\SBox){=}P(\hat{l}|\SCirc)$ whenever $ \hat{l} $ adopts a value that actually occurs in $ \VR $, namely $ 2 $ or $ 3 $. Analogous computations yield analogous results for the hypothetical experiment shown in \reffig{circularshift}{B}.

These results contrast with those found in studies of noise correlations in two ways. First, those studies employed surrogate responses generated through stimulus-dependent transformations and often contain responses with information not present in the original neural responses. Second, for those surrogate responses, we can prove that $ \diq{Q}{DL} $ cannot reach $ 100\,\% $ unless $ \dienc{\VR}{\VRT} $ does as well. Specifically, we can prove the following theorem

\begin{theorem}
When ignoring noise correlations, $ \diq{Q}{DL}{=}\infoes{\VR}$ if and only if $ \dienc{\VR}{\VRNI}{=}\infoes{\VR} $, regardless of whether stochastic codes exist that map the actual responses $ \VR $ into the surrogate responses $ \VRNI $ generated assuming noise independence.
\end{theorem}

\begin{proof}
Consider a neural code $ \VR{=}[R_1,\ldots,R_N] $ and recall that the range of $ \VRNI $ includes that of $ \VR $. Therefore, $ \diq{Q}{DL}{=}\infoes{\VR}$ implies that the minimum in \refeq{met::eq::diqdl} is attained when $ \theta{=}0 $. In that case, equation (B13a) in \cite{latham2005} yields the following $ \sum_{\Vs,r_n}{P(\Vs,r_n)\,\log P(r_n|\Vs)}{=}\sum_{\Vs,r_n}{P(\Vs)\,P(r_n)\,\log P(r_n|\Vs)}$ for all $ 1{\leq}n{\leq}N $ for all $ n $. After some more algebra and recalling that the Kullback-Leibler divergence is never negative, this equation becomes $ \infoes{R_n}{=}0 $. Consequently $ \dienc{\VR}{\VRNI}{=}\infoes{\VR} $, thereby proving the ''only if'' part. For the ''if'' part, it is sufficient to notice that the last equality implies that $ \PNI(\Vr|\Vs){=}\PNI(\Vr) $.
\end{proof}

\noindent Consequently, our results cannot be inferred from the results obtained in studies of noise correlations.

In the third group (\refsfig{generalizereducecodes}{D}~and~\refsfig{naturalstochasticcodes}{A}), $ \dienc{\VR}{\VRT}{=}\diq{}{D}{=}\diq{}{DL}$, $\dienc{\VR}{\VSTL}{=}\diq{}{LS} $, $\dienc{\VR}{\VST}{=}\diq{}{B} $, and $ \diae{\VR}{\VRT}{=}\diaq{}{B} $. We can prove that these relations arise whenever the mapping from $ \VR $ into $ \VRT $ can be described using positive-diagonal idempotent stochastic matrices \citep{hognas2011}. Specifically, we can prove the following theorem

\begin{theorem}\label{met::theo::encdecrel} Consider stimulus-independent stochastic functions $ f $ from a neural code $ \VR $ into another code $ \VRT $ which range $ \mathcal{R} $ includes that of $ \VR $, and which transition probabilities $ P(\VrT|\Vr) $ can be written as positive-diagonal idempotent right stochastic matrices with row and column indexes equally enumerating $ \mathcal{R} $. Then, $ \dienc{\VR}{\VRT}{=}\dienc{\VR}{\VRB}{=}\diq{Q}{D}{=}\diq{Q}{DL}$, $ \dienc{\VR}{\VSTL}{=}\diq{Q}{LS} $, $ \dienc{\VR}{\VST}{=}\diq{Q}{B} $, and $ \diae{\VR}{\VRT}{=}\diaq{Q}{B} $, whenever reduced codes $ \VRB $ can be devised as lossless deterministic representations of $ \VRT $.
\end{theorem}

\begin{proof}
The last condition assumes that a deterministic function $ g $ maps $ \VRT $ into $ \VRB $ such that $ \infoes{\VRT}{=}\infoes{\VRB} $. Hence, $ P(\Vs|\VrT){=}P(\Vs|\VrB) $ when $ \VrB{=}g(\VrT) $, thereby carrying both $ \VRT $ and $ \VRB $ the same amount and type of information. However, the intention of this condition is more profound: It requires that, even when speaking of $ \VRB $, the decoding-oriented measures be applied to the stochastic code $ \VRT $. This avoids the problems found in \reffig{generalizereducecodes}{B}, because the range $\mathcal{R} $ of $ \VRT $ includes that of $ \VR $. Futhermore, the restrictions on $ f $ imply that  $ \mathcal{R} $ can be partitioned into non-overlapping sets $ \mathcal{R}_1,\ldots,\mathcal{R}_M $, each of which is mapped by $ f $ onto itself, and most importantly, that $ P(\VrT|\Vr){=}P(\VrT|\mathcal{R}_m) $ when $ \Vr{\in}\mathcal{R}_m $. Hence, for $ \VrT{\in}\mathcal{R}_m$, $ P(\VrT|\Vs){=}P(\VrT|\mathcal{R}_m)\,P(\mathcal{R}_m|\Vs) $, thereby yielding $ P(\Vs|\VrT){=}P(\Vs|\mathcal{R}_m) $ and $ P(\Vs|\VrT,\theta){=}P(\Vs|\mathcal{R}_m,\theta) $. Our result follow immediately after recalling that $ P(\Vs|\VrB){=}P(\Vs|\VrT) $. 
\end{proof}

This theorem is unnecessary for the quantitative equality between encoding- and decoding-oriented measures, but ensures that the equalities hold not only in amount but also in type. For example, in \refsfig{circularshift}{B}, we found that $ \dienc{\VR}{\VRT}{=}\dienc{\VR}{\VSTL}{=}\dienc{\VR}{\VST}{=}\diq{Q}{D}{=}\diq{Q}{DL}{=}\diq{Q}{LS}{=}\diq{Q}{B}{=}50\,\% $ and $ \diae{\VR}{\VRT}{=}\diaq{Q}{B}{\approx}67\,\% $ even though theorem~\ref{met::theo::encdecrel} does not hold. However, the losses are not necessarily of the same type as we proved before. Most importantly, it formally justifies and clarifies the pervasive idea that, for reduced codes, the decoding-oriented measures reduce to their encoding-oriented counterparts. Indeed, the fourth group shows that, when not fulfilled, decoding-oriented measures can exceed their encoding-oriented counterparts (\refsfig{circularshift}{A}), or even be undefined (\refsfig{generalizereducecodes}{B}; see Methods). Consequently, contrary to previous studies \citep{nirenberg2003,quiroga2009,latham2013}, deterministic transformations need not ensure that the encoding and the decoding problems are related.

These results also reveal unexpected confounds in \cite{nirenberg2003,quiroga2009,latham2013} that invalidate their conclusions by turning otherwise decoding-oriented measures into encoding-oriented ones. To reveal and resolve the flaws in the aforementioned studies, recall the experiment in \reffig{circularshift}{A}. There, $ \VRT{=}f(\VR) $, with $ f $ being a deterministic bijective function. According to the above studies, $ \diq{Q}{D} $ should be computed through \refeq{met::eq::diqd} but with $ P(\VS{=}\Vs|\VRT{=}\Vr)$ replaced by $Q(\Vs|\Vr){\propto}P(\Vs)\,P(f(\Vr)|\Vs)$. Because $ f $ is bijective, $ Q(\Vs|\Vr){=}P(\Vs|\Vr){=}P(\VS{=}\Vs|\VRT{=}f(\Vr)) $. However, we can prove the following theorem

\begin{theorem}
When deterministic functions $ f $ exist such that $ \VRT{=}f(\VR) $, replacing $ P(\VS{=}\Vs|\VRT{=}\Vr)$ with $ P(\VS{=}\Vs|\VRT{=}f(\VR))$ in 
\refeq{met::eq::diqd} turns $ \diq{Q}{D} $ into an encoding-oriented measure.
\end{theorem}

\begin{proof}
Following the reasoning of previous studies \citep{nirenberg2001,nirenberg2003,latham2005,latham2013}, \refeq{met::eq::diqd} with the above replacement is analogous to comparing two optimal question-asking strategies induced by $ P(\Vs|\Vr) $ and $ P(\Vs|\VrT) $, respectively. The former is applied directly on $ \VR $, but the latter is not. Instead, the latter is applied on $ f(\VR){=}\VRT $. Consequently, the comparison addresses the effect of transforming $ \VR $ into $ \VRT $ before feeding them into the decoding process, thereby turning $ \diq{Q}{D} $ into an encoding-oriented measure.
\end{proof}

\noindent Hence, our theorem proves that the computations in the aforementioned studies turn $ \diq{Q}{D} $ into an encoding-oriented measure, and that their conclusions confuse encoding and decoding. Consequently, caution must be exercise when computing the necessary probabilities for all decoding-oriented measures, including but not limited to $ \diq{Q}{D} $.

In practice, our results open up the possibility of increasing the efficiency of decoders constructed with approximate descriptions of the neural responses, usually called approximate or mismatched decoders, by adding suitable amounts and types of noise to the decoder input. To see this, recall the example of \reffig{differentinformationtype}{} in which a decoder constructed with noisy data ($ \VRT $; \reffig{differentinformationtype}{A}) was employed to decode quality data ($ \VR $; \reffig{differentinformationtype}{B}), thereby causing the information losses $\diq{Q}{D}{=}\diq{Q}{DL}{=}\infoes{\VR}/2 $. These losses can be decreased by feeding the decoder with a degraded version $ \VRB $ of the quality data (\refsfig{differentinformationtype}{D}) generated through a stimulus-independent transformation that adds latency noise (\refsfig{differentinformationtype}{E}). Decoding $ \VR $ as if it was $ \VRT $ by first transforming $ \VR $ into $ \VRB $ results in $ \diq{Q}{D}{=}\diq{Q}{DL}{=}\infoes{\VR}/3 $, thereby recovering $ 33\,\% $ of the information previously lost. On the contrary, adding spike-count noise will tend to increase the losses. Thus, we have proved that adding suitable amounts and type of noise can increase the performance of approximate decoders even for response aspects beyond noise correlations. In addition, this result also indicates that, contrary to previously thought \citep{shamir2014}, decoding algorithms need not match the encoding mechanisms for performing optimally from an information-theoretical standpoint.

\section*{Conclusion}

Here we have reformulated the coding problem from the observer perspective as concerning two questions: what needs to be seen --- the encoding problem --- and how it must be interpreted --- the decoding problem. These two problems were here shown to provide limited insight into each other when studying spike-time precision, response discrimination, or other aspects of the neural response. Furthermore, we have shown that response aspects may play different roles in encoding and decoding even when defined through stimulus-independent transformations, and that decoding need not match encoding for response aspects beyond noise correlations. These findings constitute a major departure from traditional views on the role of response aspects in neural coding and brain computations. On the practical side, our most outstanding finding is that the decoded information need not be limited by the information carried in the data used to construct the decoder. This finding was here shown to open up new possibilities for increasing the performance of existing approximate decoders. Most importantly, it also indicates that decoders need not be retrained for operating optimally on data of higher quality than the one used in their construction, thereby potentially saving experimental and computational resources, and reducing the complexity and cost of neural prosthetics.

\subsection*{Acknowledgments}

This work was supported by the Ella and Georg Ehrnrooth Foundation.

\section*{Appendix}

\subsection*{Transition probabilities used in the figures}

In \reffig{generalizereducecodes}{ }, the transition probabilities $ Q(\VRT|\VR) $ are defined as follows

\begin{equation}\arraycolsep=1ex
Q(\VRT|\VR)= \left[\begin{array}{cccccc}
1 & 0 & 0 \\
0 & 1 & 0 \\
0 & 1 & 0 \\
0 & 0 & 1 
\end{array}\right]\, \mbox{ for \reffig{generalizereducecodes}{B};}\end{equation}\label{app::eq::Qfig2B}

\begin{equation}\arraycolsep=1ex
Q(\VRT|\VR)= \left[\begin{array}{cccccc}
1 &  0 & 0 & 0 \\
0 &  a & \bar{a} & 0 \\
0 &  a & \bar{a} & 0 \\
0 &  0 & 0 & 1 
\end{array}\right]\, \mbox{ for \reffig{generalizereducecodes}{C}; and}\end{equation}\label{app::eq::Qfig2C}

\begin{equation}\arraycolsep=1ex
Q(\VRT|\VR)= \left[\begin{array}{cccccc}
b & \bar{b} & 0 & 0 & 0 & 0\\
0 & 0 & C & \bar{c} & 0 & 0\\
0 & 0 & c & \bar{c} & 0 & 0\\
0 & 0 & 0 & 0 & d & \bar{d}
\end{array}\right]\, \mbox{ for \reffig{generalizereducecodes}{D}.} \end{equation}\label{app::eq::Qfig2D}

\noindent Here, we have used the following conventions: row and column indexes enumerate $ \VR$ and $ \VRT $, respectively, in lexicographical order; $ a{=}P(\VR{=}[3,1])/P(L{=}3) $; $ 0{<}b,c,d{<}1 $; and $ \bar{x}{=}1{-}x $ for any number $ x $. Notice that, the matrix in \refeq{app::eq::Qfig2C} is positive-diagonal idempotent right stochastic as in theorem~\ref{met::theo::encdecrel}, whereas the one in \refeq{app::eq::Qfig2D} can be made positive-diagonal idempotent by adding rows.

In \reffig{naturalstochasticcodes}{ }, the transition probabilities $ Q(\VRT|\VR) $ are defined using matrices in which row and column indexes enumerate $ \VR $ and $ \VRT $, respectively, according to their latencies in increasing order, thereby yielding 

\begin{equation}\arraycolsep=1ex
Q(\VRT|\VR)= \frac{1}{3} \left[\begin{array}{cccccc}
1 & 1 & 1 & 0 & 0 & 0\\
1 & 1 & 1 & 0 & 0 & 0\\
0 & 0 & 0 & 1 & 1 & 1
\end{array}\right]\, \mbox{ for \reffig{naturalstochasticcodes}{A};}\end{equation}\label{app::eq::Qfig7A}

\begin{equation}\arraycolsep=1ex
Q(\VRT|\VR)= \frac{1}{6} \left[\begin{array}{ccccc}
2 & 2 & 2 & 0 & 0   \\
0 & 2 & 2 & 2 & 0   \\
0 & 0 & 2 & 2 & 2   
\end{array}\right]\, \mbox{ for \reffig{naturalstochasticcodes}{B};}\end{equation}\label{app::eq::Qfig7B}

\begin{equation}\arraycolsep=1ex
Q(\VRT|\VR)= \frac{1}{6} \left[\begin{array}{ccc}
3   & 3   & 0   \\
2   & 2   & 2   \\
0   & 3   & 3   
\end{array}\right]\, \mbox{ for \reffig{naturalstochasticcodes}{C}; and}\end{equation}\label{app::eq::Qfig7C}

\begin{equation}\arraycolsep=1ex
Q(\VRNI|\VR)= \frac{1}{2} \,\left[\begin{array}{ccccccccc}
2 b     & \bar{b} c & \bar{c} \bar{b} & \bar{b} c & 0 & 0 & \bar{c} \bar{b} & 0 & 0\\
\bar{a} & 2 a & 0 & 0  & \bar{a} & 0 & 0 & 0 & 0\\
a & 0 & 0 & 2\bar{a} & a & 0 & 0 & 0 & 0\\
0 & b & 0 & b & 2\bar{b} c & \bar{c} \bar{b} & 0 & \bar{c} \bar{b} & 0\\
0 & 0 & b & 0 & 0 & \bar{b} c & b & \bar{b} c & 2\bar{c} \bar{b}\end{array}\right]\, \mbox{ for \reffig{naturalstochasticcodes}{D}.}\end{equation}\label{app::eq::Qfig8}

\noindent Here, $ a{=}P(\VR{=}[1,2]|\VS{=}\protect\SBox) $; $ b{=}P(\VR{=}[1,1]|\VS{=}\protect\SCirc)  $; and $ c{=}P(\VR{=}[2,2]|\VS{=}\protect\SCirc)/\bar{b} $.

In \reffig{circularshift}{ }, the transition probabilities $ Q(\VRT|\VR) $ are defined using matrices in which row and column indexes indicate to the values of $ \varPhi$ and $ \hat{\varPhi}$, respectively, thereby yielding 

\begin{equation}\arraycolsep=1ex
Q(\VRT|\VR)=  \left[\begin{array}{cccc}
0 & 1 & 0 & 0 \\
0 & 0 & 1 & 0 \\
0 & 0 & 0 & 1 \\
1 & 0 & 0 & 0 
\end{array}\right]\, \mbox{ for \reffig{circularshift}{A}; and }\end{equation}\label{app::eq::Qfig6A}

\begin{equation}\arraycolsep=1ex
Q(\VRT|\VR)=  \frac{1}{3}\,\left[\begin{array}{cccccc}
1 & 1 & 0 & 0 & 0 & 1 \\
1 & 1 & 1 & 0 & 0 & 0 \\
0 & 1 & 1 & 1 & 0 & 0 \\
0 & 0 & 1 & 1 & 1 & 0 \\
0 & 0 & 0 & 1 & 1 & 1 \\
1 & 0 & 0 & 0 & 1 & 1  
\end{array}\right]\, \mbox{ for \reffig{circularshift}{B}.}\end{equation}\label{app::eq::Qfig6B}

In \reffig{correlationdependence}{ }, the transition probabilities $ Q(\VRT|\VR) $ can be obtained by solving \refeq{met::eq::stochasticcondition} and is given by the following

\begin{equation}\arraycolsep=1ex
Q(\VRNI|\VR)= 0.5\,\left[\begin{array}{ccccccc}
\bar{a} & 2a & 0 & \bar{a} & 0  & 0 & 0\\
a & 0 &  2\bar{a} & a & 0  & 0 & 0\\
0 & 0 &  0 & b & 2\bar{b}  & 0 & b \\
0 & 0 &  0 & \bar{b}  & 0 & 2b & \bar{b}
\end{array}\right]\, ,\end{equation}\label{app::eq::Qfig9}

\noindent with row and column indexes enumerating $ \VR$ and $ \VRT $, respectively, in lexicographical order; $ a{=}P(\VR{=}[1,2]|\VS{=}\SBox) $; and $ b{=}P(\VR{=}[3,2]|\VS{=}\SCirc)  $.

In \reffig{differentinformationtype}{ }, the transition probabilities $ Q(\VRT|\VR) $ can be defined using matrices in which row and column indexes enumerate $ \VR $ and $ \VRT $, respectively, in lexicographical order, thereby yielding

\begin{equation}\arraycolsep=1ex
Q(\VRT|\VR)=  \frac{1}{9}\,\left[\begin{array}{*{20}{c@{\hspace{1ex}}}}
1 & 1 & 1 & 0 & 1 & 1 & 1 & 0 & 1 & 1 & 1 & 0 & 0 & 0 & 0 & 0 & 0 & 0 & 0 & 0 \\
0 & 1 & 1 & 1 & 0 & 1 & 1 & 1 & 0 & 1 & 1 & 1 & 0 & 0 & 0 & 0 & 0 & 0 & 0 & 0 \\
0 & 0 & 0 & 0 & 0 & 0 & 0 & 0 & 1 & 1 & 1 & 0 & 1 & 1 & 1 & 0 & 1 & 1 & 1 & 0 \\
0 & 0 & 0 & 0 & 0 & 0 & 0 & 0 & 0 & 1 & 1 & 1 & 0 & 1 & 1 & 1 & 0 & 1 & 1 & 1
\end{array}\right]\, \mbox{ for \reffig{differentinformationtype}{C}; and }\end{equation}\label{app::eq::Qfig10C}

\begin{equation}\arraycolsep=1ex
Q(\VRB|\VR)=  \frac{1}{3}\,\left[\begin{array}{*{20}{c@{\hspace{1ex}}}}
0 & 0 & 0 & 0 & 1 & 1 & 1 & 0 & 0 & 0 & 0 & 0 & 0 & 0 & 0 & 0 & 0 & 0 & 0 & 0 \\
0 & 0 & 0 & 0 & 1 & 1 & 1 & 0 & 0 & 0 & 0 & 0 & 0 & 0 & 0 & 0 & 0 & 0 & 0 & 0 \\
0 & 0 & 0 & 0 & 0 & 0 & 0 & 0 & 0 & 0 & 0 & 0 & 1 & 1 & 1 & 0 & 0 & 0 & 0 & 0 \\
0 & 0 & 0 & 0 & 0 & 0 & 0 & 0 & 0 & 0 & 0 & 0 & 1 & 1 & 1 & 0 & 0 & 0 & 0 & 0 
\end{array}\right]\, \mbox{ for \reffig{differentinformationtype}{E}.}\end{equation}\label{app::eq::Qfig10D}

\bibliographystyle{apa}

\end{document}